\documentclass[a4paper,11pt]{article}

\usepackage{amssymb}
\usepackage{amsthm}

\usepackage{lmodern}
\usepackage{enumitem}
\usepackage{url}
\usepackage{comment}
\usepackage{amsmath}
\usepackage{fullpage}
\usepackage[numbers,sort&compress]{natbib}

\def\PDFcompatibleLineBreak{\texorpdfstring{\\}{}}

\usepackage{hyperref}

\usepackage{environ}

\usepackage{thmtools, thm-restate, xpatch}
\makeatletter
\xpatchcmd{\thmt@restatable}
{\csname #2\@xa\endcsname\ifx\@nx#1\@nx\else[{#1}]\fi}
{\csname #2\@xa\endcsname\ifthmt@thisistheone\else[{restated}]\fi}
\makeatother

\newtheorem{theorem}{Theorem}
\newtheorem{lemma}{Lemma}
\newtheorem{observation}{Observation}
\newtheorem{definition}{Definition}
\newtheorem{proposition}{Proposition}
\newtheorem{claim}{Claim}

\begin{document}

\title{Parameterized Algorithms for Generalizations of Directed Feedback Vertex Set\thanks{A preliminary version of these results appeared in the Proceedings of the 11th International Conference on Algorithms and Complexity~\cite{GokeMM2019}.}}

\author{Alexander G{\"o}ke \thanks{TU Hamburg, Institute for Algorithms and Complexity, Hamburg, Germany. \texttt{alexander.goeke@tuhh.de}. Supported by DFG grant MN 59/1-1.}
  \and D{\'a}niel Marx\thanks{Max-Planck-Institut f{\"u}r Informatik, Saarbr{\"u}cken, Germany. \texttt{dmarx@mpi-inf.mpg.de}. Supported by ERC Consolidator Grant SYSTEMATICGRAPH (755978).}
  \and Matthias Mnich\thanks{TU Hamburg, Institute for Algorithms and Complexity, Hamburg, Germany. \texttt{matthias.mnich@tuhh.de}. Supported by DFG grant MN 59/4-1.}
  }

\maketitle

\begin{abstract}
  The {\sc Directed Feedback Vertex Set} (DFVS) problem takes as input a directed graph~$G$ and seeks a smallest vertex set~$S$ that hits all cycles in $G$.
  This is one of Karp's~21 $\mathsf{NP}$-complete problems.
  Resolving the parameterized complexity status of DFVS was a long-standing open problem until Chen et al. [STOC 2008, J. ACM 2008] showed its fixed-parameter tractability via a $4^kk! n^{\mathcal{O}(1)}$-time algorithm, where $k = |S|$.

  Here we show fixed-parameter tractability of two generalizations of DFVS:
  \begin{itemize}
    \item Find a smallest vertex set $S$ such that every strong component of $G - S$ has size at most~$s$: we give an algorithm solving this problem in time $4^k(ks+k+s)!\cdot n^{\mathcal{O}(1)}$.
      This generalizes an algorithm by Xiao [JCSS 2017] for the undirected version of the problem.
    \item Find a smallest vertex set $S$ such that every non-trivial strong component of $G - S$ is 1-out-regular: we give an algorithm solving this problem in time $2^{\mathcal{O}(k^3)}\cdot n^{\mathcal{O}(1)}$.
  \end{itemize}
  We also solve the corresponding arc versions of these problems by fixed-parameter algorithms.
  
\end{abstract}


\section{Introduction}
\label{sec:introduction}
The {\sc Directed Feedback Vertex Set} (DFVS) problem is that of finding a smallest vertex set~$S$ in a given digraph $G$ such that $G - S$ is a directed acyclic graph.
This problem is among the most classical problems in algorithmic graph theory.
It is one of the 21 $\mathsf{NP}$-complete problems on Karp's famous list~\cite{Karp1972}.

Consequently, the DFVS problem has long attracted researchers in approximation algorithms.
The current best known approximation factor that can be achieved in polynomial time for $n$-vertex digraphs with optimal fractional solution value $\tau^*$ is\linebreak $\mathcal {O}(\min\{\log\tau^*\log\log\tau^*, \log n\log\log n\})$ due to Seymour~\cite{Seymour1995}, Even et al.~\cite{EvenEtAl1998} and Even et al.~\cite{EvenEtAl2000}.
On the negative side, Karp's $\mathsf{NP}$-hardness reduction shows the problem to be $\mathsf{APX}$-hard, which rules out the existence of a polynomial-time approximation scheme (PTAS) assuming $\mathsf{P}\not=\mathsf{NP}$.
Assuming the Unique Games Conjecture, the DFVS problem does not admit a polynomial-time $\mathcal{O}(1)$-approximation~\cite{GuruswamiEtAl2011,GuruswamiLee2016,Svensson2013}.

The DFVS problem has also received a significant amount of attention from the perspective of parameterized complexity.
The main parameter of interest there is the optimal solution size $k = |S|$.
The problem can easily be solved in time $n^{\mathcal{O}(k)}$ by enumerating all $k$-sized vertex subsets $S\subseteq V(G)$ and then seeking a topological order of $G - S$.
The interesting question is thus whether the DFVS problem is \emph{fixed-parameter tractable} with respect to~$k$, which is to devise an algorithm with run time $f(k)\cdot n^{\mathcal{O}(1)}$ for some computable function~$f$ depending only on $k$.
It was a long-standing open problem whether DFVS admits such an algorithm.
The question was finally resolved by Chen et al. who gave a $4^k k!k^4\cdot \mathcal{O}(nm)$-time algorithm for graphs with $n$ vertices and $m$ arcs.
Recently, an algorithm for DFVS with run time $4^kk! k^5\cdot \mathcal{O}(n+m)$ was given by Lokshtanov et al.~\cite{LokshtanovEtAl2018}.
It is well-known that the \emph{arc} deletion variant is parameter-equivalent to the \emph{vertex} deletion variant and hence {\sc Directed Feedback Arc Set (DFAS)} can also be solved in time $4^kk! k^5\cdot \mathcal{O}(n+m)$.

Once the breakthrough result for DFVS was obtained, the natural question arose how much further one can push the boundary of (fixed-parameter) tractability.
On the one hand, Chitnis et al.~\cite{ChitnisEtAl2015} showed that the generalization of DFVS where one only wishes to hit cycles going through a specified subset of nodes of a given digraph is still fixed-parameter tractable when parameterized by solution size.
On the other hand, Lokshtanov et al.~\cite{LokshtanovEtAl2020} showed that finding a smallest set of vertices of hitting only the \emph{odd} directed cycles of a given digraph is $\mathsf{W}[1]$-hard, and hence not fixed-parameter tractable unless $\mathsf{FPT} = \mathsf{W}[1]$.

\medskip
\noindent
\textbf{Our contributions.}
For another generalization the parameterized complexity is still open:
In the {\sc Eulerian Strong Component Arc (Vertex) Deletion} problem, one is given a directed multigraph~$G$, and asks for a set~$S$ of at most~$k$ vertices such that every strong component of $G- S$ is Eulerian, that is, every vertex has the same in-degree and out-degree within its strong component.
The arc version of this problem was suggested by Cechl\'{a}rov\'{a} and Schlotter~\cite{CechlarovaSchlotter2011} in the context of housing markets.
Marx~\cite{Marx2012} explicitly posed determining the parameterized complexity of {\sc Eulerian Strong Component Vertex Deletion} as an open problem.
Notice that these problems generalize the DFAS/DFVS problems, where each strong component of $G - S$ has size one and thus is Eulerian.

\begin{restatable}{theorem}{eulerianstrongcomponentvdishard}
\label{thm:eulerianscvd_main}
  {\sc Eulerian Strong Component Vertex Deletion} is $\mathsf{NP}$-hard and \mbox{$\mathsf{W[1]}$-hard} parameterized by solution size $k$, even for $(k+1)$-strong digraphs.
\end{restatable}

Alas, we are unable to determine the parameterized complexity of {\sc Eulerian Strong Component Arc Deletion}, which appears to be more challenging.
Hence, we consider two natural generalizations of DFAS which may help to gain better insight into the parameterized complexity of that problem.

\medskip
First, we consider the problem of deleting a set of $k$ arcs or vertices from a given digraph such that every strong component has size at most~$s$.
Thus, the {\sc DFAS/DFVS} problems corresponds to the special case when $s = 1$.
Formally, the problem {\sc Bounded Size Strong Component Arc (Vertex) Deletion} takes as input a multi-digraph $G$ and integers $k,s$, and seeks a set~$S$ of at most~$k$ arcs or vertices such that every strong component of $G - S$ has size at most $s$.

The \emph{undirected} case of {\sc Bounded Size Strong Component Vertex (Arc) Deletion} was studied recently.
There, one wishes to delete at most~$k$ vertices of an undirected $n$-vertex graph such that each connected component of the remaining graph has size at most $s$.
For $s$ being constant, Kumar and Lokshtanov~\cite{KumarLokshtanov2016} obtained a kernel of size $2sk$ that can be computed in $n^{\mathcal{O}(s)}$ time; note that the degree of the run time in the input size $n$ depends on $s$ and is thus not a fixed-parameter algorithm.
For general $s$, there is a $9sk$-sized kernel computable in time $\mathcal{O}(n^4m)$ by Xiao~\cite{Xiao2017}.
The directed case---which we consider here---generalizes the undirected case by replacing each edge by arcs in both directions.

Our main result here is to solve the directed case of the problem by a fixed-parameter algorithm:
\begin{theorem}
\label{thm:boundedsizedeletion_main}
  There is an algorithm that solves {\sc Bounded Size Strong Component Arc (Vertex) Deletion} in time $4^k(ks + k + s)!\cdot n^{\mathcal{O}(1)}$ for $n$-vertex multi-digraphs and integers~$k,s$.
\end{theorem}
In particular, our algorithm exhibits the same asymptotic dependence on~$k$ as does the algorithm by Chen et al.~\cite{ChenEtAl2008} for the DFVS/DFAS problem, which corresponds to the special case $s = 1$.

Another motivation for this problem comes from the $k$-linkage problem, which asks for $k$ pairs of terminal vertices in a digraph if they can be connected by $k$ mutually arc-disjoint paths.
The $k$-linkage problem is $\mathsf{NP}$-complete already for $k = 2$~\cite{FortuneEtAl1980}.
Recently, Bang-Jensen and Larsen~\cite{BangJensenLarsen2016} solved the $k$-linkage problem in digraphs where strong components have size at most~$s$.
Thus, finding induced subgraphs with strong components of size at most $s$ can be of interest in computing $k$-linkages.

\medskip

Our second problem 
is that of deleting a set of~$k$ arcs or vertices from a given digraph such that each remaining strong component is \emph{$r_C$-out-regular}, meaning that every vertex has out-degree exactly $r_C$ in its strong component $C$, for $r_C \le 1$.
So in particular, every strong component is Eulerian, as in the {\sc Eulerian Strong Component Arc Deletion} problem.
Observe that in the {\sc DFAS/DFVS} problem we delete~$k$ arcs or vertices from a given directed graph such that each remaining strong component is 0-out-regular (trivial).
Formally, we consider the {\sc 1-Out-Regular Arc (Vertex) Deletion} problem in which for a given multi-digraph~$G$ and integer~$k$, we seek a set $S$ of at most $k$ arcs (vertices) such that every component~$C$ of $G - S$ is $r_C$-out-regular with $r_C \in \lbrace 0,1\rbrace$.
Note that this problem is equivalent to deleting a set~$S$ of at most $k$ arcs (vertices) such that every non-trivial (consisting of more than one vertex) strong component of $G - S$ is an induced directed cycle.
In contrast to {\sc Eulerian Strong Component Vertex Deletion}, the {\sc 1-Out-Regular Arc (Vertex) Deletion} problem \emph{is} monotone, in that every superset of a solution is again a solution: if we delete an additional arc or vertex that breaks a strong component that is an induced cycle into several strong components, then each of these newly created strong components is trivial.

Our result for this problem reads as follows:
\begin{theorem}
\label{thm:2regulareuleriandeletion_main}
  There is an algorithm solving {\sc 1-Out-Regular Arc (Vertex) Deletion} in time $2^{\mathcal{O}(k^3)}\cdot \mathcal{O}(n^4)$ for $n$-vertex digraphs $G$ and parameter~$k\in\mathbb N$.
\end{theorem}

%

Notice that for {\sc Bounded Size Strong Component Arc (Vertex) Deletion} and {\sc 1-Out-Regular Arc (Vertex) Deletion}, there are infinitely many instances for which solutions are arbitrarily smaller than those for DFAS (DFVS), and for any instance they are never larger.
Therefore, our algorithms strictly generalize the one by Chen et al.~\cite{ChenEtAl2008} for DFAS (DFVS).
As a possible next step towards resolving the parameterized complexity of {\sc Eulerian Strong Component Arc Deletion}, one may generalize our algorithm for {\sc 1-Out-Regular Arc Deletion} to {\sc $r$-Out-Regular Arc Deletion} for arbitrary~$r$.

\medskip
We give algorithms for vertex deletion variants only, and then reduce the arc deletion variants to them.

\clearpage
\pagebreak
\section{Notions and Notations}
\label{sec:notionsandnotations}
We consider finite directed graphs (or digraphs) $G$ with vertex set $V(G)$ and arc set $A(G)$.
We allow multiple arcs and arcs in both directions between the same pairs of vertices.
For each vertex $v\in V(G)$, its \emph{out-degree} in~$G$ is the number $d^+_G(v)$ of arcs of the form $(v,w)$ for some $w\in V(G)$, and its \emph{in-degree} in $G$ is the number $d^-_G(v)$ of arcs of the form $(w,v)$ for some $w\in V(G)$.
A vertex $v$ is \emph{balanced} if $d^+_G(v) = d^-_G(v)$.
A digraph $G$ is \emph{balanced} if every vertex $v \in V(G)$ is balanced.

For each subset $V'\subseteq V(G)$, the subgraph induced by $V'$ is the graph~$G[V']$ with vertex set~$V'$ and arc set $\{(u,v)\in A(G)~|~u,v\in V'\}$.
For any set $X$ of arcs or vertices of $G$, let $G - X$ denote the subgraph of $G$ obtained by deleting the elements of~$X$ from $G$.
For subgraphs~$G'$ of~$G$ and vertex sets $X\subseteq V(G)$ let $R^+_{G'}(X)$ denote the set of vertices that are \emph{reachable} from~$X$ in $G'$, i.e. vertices to which there is a path from some vertex in $X$.
For an $s$-$t$-walk $P$ and a $t$-$q$-walk $R$ we denote by $P \circ R$ the \emph {concatenation} of these paths, i.e. the $s$-$q$-walk resulting from first traversing $P$ and then $R$.

Let $G$ be a digraph.
Then $G$ is \emph{$1$-out-regular} if every vertex has out-degree exactly $1$.
Further,~$G$ is called \emph{strong} if either $G$ consists of a single vertex (then~$G$ is called \emph{trivial}), or for any distinct $u,v\in V(G)$ there is a directed path from $u$ to $v$.
A \emph{strong component} of $G$ is an inclusion-maximal strong induced subgraph of $G$.
Also, $G$ is \emph{$t$-strong} for some $t\in\mathbb N$ if for any $X \subseteq V(G)$ with $|X| < t$, $G - X$ is strong.
We say that $G$ is \emph{weakly connected} if its underlying undirected graph~$\langle G\rangle$ is connected.
Finally, $G$ is \emph{Eulerian} if there is a closed walk in $G$ using each arc exactly once.

\begin{definition}
  For disjoint non-empty vertex sets $X,Y$ of a digraph $G$, an arc or vertex set~$S$ is an \emph{$X \to Y$-separator} if~$S$ is disjoint from $X\cup Y$ and there is no path from~$X$ to $Y$ in $G - S$.
  
  An $X \to Y$-separator $S$ is \emph{minimal} if no proper subset of $S$ is an $X \to Y$-separator.
  An $X \to Y$-separator~$S$ is \emph{important} if there is no $X \to Y$-separator $S'$ with $|S'|\leq |S|$ and $R^+_{G - S}(X) \subset R^+_{G - S'}(X)$.
\end{definition}

\begin{proposition}[\cite{ChitnisEtAl2013}]
\label{thm:fptenumerationofimportantseperators}
  Let $G$ be a digraph and let $X,Y\subseteq V(G)$ be disjoint non-empty vertex sets.
  For every $p\geq 0$ there are at most $4^p$ important $X \to Y$-separators of size at most $p$, all of which can be enumerated in time $4^p\cdot n^{\mathcal{O}(1)}$.
\end{proposition}

\section{Tools for Generalized DFVS/DFAS Problems}
\label{sec:generalsteps}

\medskip
\noindent
\textbf{Iterative Compression.}
We use the standard technique of iterative compression.
For this, we label the vertices of the input digraph $G$ arbitrarily by $v_1,\hdots,v_n$, and set $G_i = G[\{v_1,\hdots,v_i\}]$.
%
We start with~$G_1$ and the solution $S_1 = \{v_1\}$.
As long as $|S_i| < k$, we can set $S_{i+1} = S_i \cup \{v_{i+1}\}$ and continue.
As soon as $|S_i| = k$, the set $T_{i+1} = S_i \cup \{v_{i+1}\}$ is a solution for $G_{i+1}$ of size $k+1$.
The \emph{compression variant} of our problem then takes as input a digraph~$G$ and a solution $T$ of size $k+1$, and seeks a solution $S$ of size at most $k$ for~$G$ or decides that none exists.

We call an algorithm for the compression variant on $(G_{i+1}, T_{i+1})$ to obtain a solution~$S_{i+1}$ or find out that $G_{i+1}$ does not have a solution of size $k$, but then neither has~$G$.
By at most~$n$ calls to this algorithm we can deduce a solution for the original instance $(G_n = G,k)$.


\medskip
\noindent
\textbf{Disjoint solution.}
Given an input $(G,T)$ to the compression variant, the next step is to ask for a solution $S$ for $G$ of size at most $k$ that is disjoint from the given solution $T$ of size $k+1$.
This assumption can be made by guessing the intersection $T' = S\cap T$, and deleting those vertices from $G$.
Since $T$ has $k+1$ elements, this step creates $2^{k+1}$ candidates $T'$.
The \emph{disjoint compression variant} of our problem then takes as input a graph $G - T'$, a solution $T\setminus T'$ of size $k+1-|T'|$, and seeks a solution $S'$ of size at most $k - |T'|$ disjoint from $T\setminus T'$.



\medskip
\noindent
\textbf{Covering the shadow of a solution.}
The ``shadow'' of a solution $S$ is the set of those vertices that are disconnected from $T$ (in either direction) after the removal of $S$.
A common idea of several fixed-parameter algorithms on digraphs is to first ensure that there is a solution whose shadow is empty, as finding such a shadowless solution can be a significantly easier task.
A generic framework by Chitnis et al.~\cite{ChitnisEtAl2015} shows that for special types of problems as defined below, one can invoke the random sampling of important separators technique and obtain a set $Z$ which is disjoint from a minimum solution and covers its shadow, i.e. the shadow is contained in $Z$.
What one does with this set, however, is problem-specific.
Typically, given such a set, one can use (some problem-specific variant of) the ``torso operation'' to find an equivalent instance that has a shadowless solution.
Therefore, one can focus on the simpler task of finding a shadowless solution or more precisely, finding any solution under the guarantee that a shadowless solution exists.


\begin{definition}[shadow]
  Let $G$ be a digraph and let $T,S\subseteq V(G)$.
  A vertex $v\in V(G)$ is \emph{in the forward shadow $f_{G,T}(S)$ of $S$ (with respect to $T$)} if $S$ is a $T \to \{v\}$-separator in $G$, and~$v$ is \emph{in the reverse shadow $r_{G,T}(S)$ of $S$ (with respect to $T$)} if $S$ is a $\{v\} \to T$-separator in $G$.
    
  A vertex is \emph{in the shadow of $S$} if it is in the forward or reverse shadow of $S$.
\end{definition}
Note that $S$ itself is not in the shadow of $S$ by definition of separators.

\begin{definition}[$T$-connected and $\mathcal F$-transversal]
  Let $G$ be a digraph, let $T\subseteq V(G)$ and let $\mathcal F$ be a set of subgraphs of $G$.
  We say that $\mathcal F$ is \emph{$T$-connected} if for every $F\in\mathcal F$, each vertex of~$F$ can reach some and is reachable by some (maybe different) vertex of $T$ by a walk completely contained in $F$.
  For a set~$\mathcal F$ of subgraphs of $G$, an \emph{$\mathcal F$-transversal} is a set of vertices that intersects the vertex set of every subgraph in $\mathcal F$.
\end{definition}

\noindent
Chitnis et al.~\cite{ChitnisEtAl2015} show how to deterministically cover the shadow of $\mathcal F$-transversals:
\begin{proposition}[deterministic covering of the shadow, \cite{ChitnisEtAl2015}]
\label{thm:deterministiccoveringoftheshadow}
  Let $T\subseteq V(G)$.
  In time $2^{\mathcal O(k^2)}\cdot n^{\mathcal{O}(1)}$ one can construct $t \leq 2^{\mathcal{O}(k^2)}\log^2n$ sets $Z_1,\hdots,Z_t$ such that for any set of subgraphs~$\mathcal F$ which is $T$-connected, if there exists an $\mathcal F$-transversal of size at most~$k$ then there is an $\mathcal F$-transversal $S$ of size at most~$k$ that is disjoint from $Z_i$ and such that $Z_i$ covers the shadow of $S$, for some~$i \leq t$.
\end{proposition}

\section{Hardness of Vertex Deletion}
\label{sec:eulerianscvd}
In this section we prove Theorem~\ref{thm:eulerianscvd_main}, by showing $\mathsf{NP}$-hardness and $\mathsf{W}[1]$-hardness of the {\sc Eulerian Strong Components Vertex Deletion} problem.
Before the hardness proof we recall an equivalent characterization of Eulerian digraphs:
\begin{lemma}[folklore]
\label{thm:equivalenteuleriandefinitions}
  Let $G$ be a weakly connected digraph.
  Then $G$ is Eulerian if and only if~$G$ is balanced.
\end{lemma}

We can now state the hardness reduction, which relies on the hardness of the following problem introduced by Cygan et al.~\cite{CyganEtAl2014}.
In {\sc Directed Balanced Vertex Deletion}, one is given a directed multigraph $G$ and an integer $k\in\mathbb N$, and seeks a set $S$ of at most $k$ vertices such that $G - S$ is balanced.


\begin{proposition}[\cite{CyganEtAl2014}]
\label{thm:hardnessbalancednodedeletion}
  {\sc Directed Balanced Vertex Deletion} is $\mathsf{NP}$-hard and $\mathsf{W}[1]$-hard with parameter $k$.
\end{proposition}

We will prove the hardness of {\sc Eulerian Strong Component Vertex Deletion} for $(k+1)$-strong digraphs by adding vertices ensuring this connectivity.

\eulerianstrongcomponentvdishard*
\begin{proof}
  We give a polynomial reduction from {\sc Directed Balanced Vertex Deletion}.
  Let $(G, k)$ an instance of {\sc Directed Balanced Vertex Deletion}.
  Let $G'$ arise from $G$ by adding vertices $z_1, \hdots, z_{k+1}$ and arcs $(z_i, v), (v, z_i)$ for all $v \in V(G)$ and all $i \in \{1, \hdots, k+1\}$.
  This construction obviously can be made to run in polynomial time.
  Moreover, $G'$ is $(k+1)$-strong as one needs to delete at least all $z_i$ to disconnect two vertices.
  All we have to show is that $(G,k)$ has a solution as instance of {\sc Directed Balanced Vertex Deletion} if and only if $(G', k)$ has a solution as instance of {\sc Eulerian Strong Component Vertex Deletion}.
	
  Let $S'$ be a solution to $(G', k)$ as instance of {\sc Eulerian Strong Components Vertex Deletion}.
  As $G'$ is $(k+1)$-strong, $G' - S'$ is strong.
  Moreover $S'$ is a solution, so $G' - S'$ is Eulerian (because it is the only strong component).
  Therefore, by Lemma~\ref{thm:equivalenteuleriandefinitions} every vertex of $G' - S'$ is balanced.
  Deleting the remaining vertices of $\{z_1, \hdots, z_{k+1}\}$ does not harm the balance of the remaining vertices, as for each $v \in V(G)$ and $z_i$ we delete one outgoing and one incoming arc of $v$.
  Thus $G' - (S' \cup \{z_1, \hdots, z_{k+1}\}) = G - (S' \setminus \{z_1, \hdots, z_{k+1}\})$ is balanced.
  Hence, $S' \setminus \{z_1, \hdots, z_{k+1}\}$ is a solution to $(G,k)$ as instance of {\sc Directed Balanced Vertex Deletion}.
	
  Let $S$ be a solution to $(G, k)$ as instance of {\sc Directed Balanced Vertex Deletion}.
  Then $G - S$ is balanced and by construction $G' - S$ as well.
  Furthermore, $G' -S$ is strong, and thus by Lemma~\ref{thm:equivalenteuleriandefinitions} also Eulerian.
  Hence, the only strong component of $G' - S$ is Eulerian and therefore~$S$ is a solution to $(G',k)$ as instance of {\sc Eulerian Strong Component Vertex Deletion}.
\end{proof}

\section{Bounded Size Strong Component Arc (Vertex) Deletion}
\label{sec:boundedsize}

In this section we show a fixed-parameter algorithm for the vertex deletion variant of {\sc Bounded Size Strong Component Vertex Deletion}.

\medskip
We give an algorithm that, given an $n$-vertex digraph $G$ and integers $k,s$, decides in time $4^k(ks+k+s)!\cdot n^{\mathcal{O}(1)}$ if $G$ has a set~$S$ of at most $k$ vertices such that every strong component of $G - S$ has size at most $s$.
Such a set $S$ will be called a \emph{solution} of the instance $(G,k,s)$.

The algorithm first executes the general steps ``Iterative Compression'' and ``Disjoint Solution''; it continues with a reduction to a skew separator problem.

\medskip
\noindent
\textbf{Reduction to Skew Separator Problem}
Now the goal is, given a digraph~$G$, integers $k,s \in\mathbb N$, and a solution $T$ of $(G,k+1,s)$, to decide if $(G,k,s)$ has a solution $S$ that is disjoint from $T$.
We solve this problem---which we call {\sc Disjoint Bounded Size Strong Component Vertex Deletion Reduction}---by reducing it to finding a small ``skew separator'' in one of a bounded number of reduced instances.
\begin{definition}
  Let $G$ be a digraph, and let $\mathcal X = (X_1,\hdots,X_t),\mathcal Y = (Y_1,\hdots,Y_t)$ be two ordered collections of $t\geq 1$ vertex subsets of $G$.
  A \emph{skew separator} $S$ for $(G,\mathcal X,\mathcal Y)$ is a vertex subset of $V(G)\setminus \bigcup_{i=1}^t(X_i\cup Y_i)$ such that for any index pair $(i,j)$ with $t\geq i\geq j\geq 1$, there is no path from~$X_i$ to $Y_j$ in the graph $G - S$.
\end{definition}
This definition gives rise to the {\sc Skew Separator} problem, which for a digraph~$G$, ordered collections $\mathcal X, \mathcal Y$ of vertex subsets of $G$, and an integer $k\in\mathbb N$, asks for a skew separator for $(G, \mathcal X, \mathcal Y)$ of size at most $k$.
Chen et al.~\cite{ChenEtAl2008} showed:
\begin{proposition}[{\cite[Thm. 3.5]{ChenEtAl2008}}]
\label{thm:skewsepfpt}
  There is an algorithm solving {\sc Skew Separator} in time $4^kk\cdot\mathcal O(n^3)$ for $n$-vertex digraphs $G$.
\end{proposition}

The reduction from {\sc Disjoint Bounded Size Strong Component Vertex Deletion Reduction} to {\sc Skew Separator} is as follows.
As $T$ is a solution of $(G,k+1,s)$, we can assume that every strong component of $G - T$ has size at most $s$.
Similarly, we can assume that every strong component of $G[T]$ has size at most~$s$, as otherwise there is no solution~$S$ of $(G,k,s)$ that is disjoint from $T$.
Let $\{t_1,\hdots,t_{k+1}\}$ be a labeling of the vertices in $T$.

\begin{lemma}
\label{thm:constructcandidatesets}
  There is an algorithm that, given an $n$-vertex digraph $G$, integers $k,s\in\mathbb N$, and a solution~$T$ of $(G,k+1,s)$, in time $\mathcal{O}((ks+s-1)! )\cdot n^{\mathcal{O}(1)}$ computes a collection~$\mathcal C$ of at most $(ks+s-1)!$ vectors $C = (C_1,\hdots,C_{k+1})$ of length $k+1$, where $t_h\in C_h\subseteq V(G)$ for $h = 1,\hdots,k+1$, such that for some solution $S$ of $(G,k,s)$ disjoint from $T$, there is a vector $C\in \mathcal C$ such that the strong component of $G - S$ containing $t_h$ is exactly $G[C_h]$ for $h = 1,\hdots,k+1$.
\end{lemma}
\begin{proof}
  Fix a hypothetical solution $S$ of $(G,k,s)$ that is disjoint from $T$.
  The algorithm computes, for each vertex $t_h\in T$, a set $C_h\ni t_h$ of at most $s$ vertices such that $C_h$ induces a strong component of $G - S$.
  (These sets $C_h$ must exist as~$S$ is required to be disjoint from~$T$.)
  Notice that the definition of~$C_h$ must only depend on~$t_h$ but not on $S$.
  Vertices in~$C_h$ (other than $t_h$) may or may not belong to $T$ and in particular it can be that $t_{h'}\in C_h$ for some $h'\in\{1,\hdots,k+1\}\setminus \{h\}$.
  Thus, for distinct $t_h,t_{h'}\in T$, sets~$C_h$ and~$C_{h'}$ possibly overlap.
	
  Intuitively, to compute $C_h$ define a ``candidate vertex for $t_h$'' as a vertex $u\in V(G)\setminus\{t_h\}$ that can potentially belong to the same strong component of $G - S$ as $t_h$, for a hypothetical solution~$S$ of size at most $k$ that is disjoint from $T$.
  We want to bound the number of candidate vertices for each $t_h\in T$, or, more precisely, the number of ``candidate sets'' $C_h$ for which $C_h\ni t_h$ can be exactly the vertex set of the strong component that contains~$t_h$, after deleting a set $S\subseteq V(G)\setminus T$ of at most $k$ vertices from $G$.
	
  Formally, the algorithm constructs sets $C_h$ iteratively by a simple branching algorithm along the following lines.
  It starts with an initial set $C_h^0 = \{t_h\}$ and a guessed set $S = \emptyset$.
  For $i\geq 0$, suppose that it has already constructed a set $C_h^i$ that must be a subset of $C_h$, and we want to either extend~$C_h^i$ to a proper superset $C_h^{i+1}$ or decide that $C_h = C_h^i$.
  If there is a path $P$ in $G$ of length at least two and length at most $s-|C_h^i|$ whose both end vertices are in $C_h^i$ and whose internal vertices are all outside~$C_h^i$, then it branches into two cases:
  \begin{itemize}
    \item either some vertex $u$ of $P$ belongs to the deletion set $S$ (meaning that we add $u$ to $S$),
	\item or the entire path $P$ belongs to the candidate set $C_h^{i+1}$ (meaning that we add the set $V^\circ(P)$ of all internal vertices of $P$ to $C_h^i$). 
  \end{itemize}
  Thus, in each branching step, either the size of $S$ strictly increases, or the size of $C_h^i$ strictly increases.
  Note that the size of~$S$ is bounded by $k$, and the size of $C_h^i\subseteq C_h$ is bounded by $s$.
  Hence, in the first branch, adding~$u$ to $S$ implies that the budget $k - |S|$ strictly decreases to $k - |S\cup\{u\}|$, whereas in the second branch, adding $V^\circ(P)$ to $C_h^i$ strictly decreases the budget $s - |C_h^i|$ to $s - |C_h^i\cup V^\circ(P)|$.
  We repeat this branching until the size of $S$ reaches the limit of~$k$ or the size of~$C_h^i$ reaches the limit of $s$, or if there are no paths left of length at most $s - |C_h^i|$ with both end vertices inside $C_h^i$ and all internal vertices outside~$C_h^i$.
  At this point, the set~$C_h^i$ will not be further extended, and $C_h := C_h^i$ is a candidate set for~$t_h$.
  This completes the algorithm description.
	
  We analyze the run time of the algorithm.
  To construct all possible vectors\linebreak $C = (C_1,\hdots,C_{k+1})\in \mathcal C$, the search tree that arises from the branching has a number of leaves that is bounded by a function $f$ of $k$ and $q$ only, where $q = \sum_{h=1}^{k+1}(s-|C_h^i|)$ is the sum of the remaining capacities of the $C_h$'s.
  By the above branching, this function satisfies the recursion $f(k,q)\leq (s-|C_h^i|)f(k-1,q) + f(k,q-1)$, as in the first branch there are at most $s-|C_h^i|$ choices for vertex $u$ each of which reduces the budget of $k - |S|$ by 1, and one branch which reduces the budget of $q$ by the number of internal vertices of $P$ which is at least~$1$.

  To obtain an upper bound on the growth of $f$, we first notice that $f(0,q) = 1$ for all $q\in\mathbb N$ and $f(k,0) = 1$ for all $k\geq 0$.
  We then claim that $f(k,q)\leq (q + k)!$, since by induction for $k,q\in\mathbb N$ it holds
  \begin{eqnarray*}
    f(k,q) & \leq & (s - |C_i^h|)f(k-1,q) + f(k,q-1)\\
           & \leq & (s - |C_i^h|)(q+k-1)! + (q-1+k)!\\
           & =    & (s - |C_i^h| + 1)(q+k-1)!\\
		   & =    & \left(\frac{s - |C^i_h| + 1}{q+k}\right)(q+k)!\\
		   & \leq &   (q+k)!,
  \end{eqnarray*}
  where in the last inequality we used that $q \geq s - |C_i^h|$ and $k\geq 1$.
  Hence, the search tree has at most $(q+k)!$ leaves, each leaf corresponding to some vector $C\in\mathcal C$.
  The initial capacity $q$ satisfies $q = (k+1)(s-1)$, and thus $|\mathcal C|\leq (ks+s-1)!$.
  Since each branching step can be executed in polynomial time, the search tree (and hence the set $\mathcal C$) can be constructed in time $(q+k)!\cdot n^{\mathcal{O}(1)}$.
  Thus, the overall run time is $(q+k)!\cdot n^{\mathcal O(1)} = (ks + s - 1)!\cdot n^{\mathcal O(1)}$.
%
\end{proof}
Observe that for each vertex $t_h\in T$ each set $C_h$ contains at most $s$ vertices, and together with the run time of the algorithm in Lemma~\ref{thm:constructcandidatesets} directly implies that $\mathcal C$ contains at most $(ks + s - 1)!$ vectors.

Armed with Lemma~\ref{thm:constructcandidatesets}, we can hence restrict our search for a solution $S$ of $(G,k,s)$ disjoint from~$T$ to those $S$ that additionally are ``compatible'' with a vector in $\mathcal C$.
Formally, a solution~$S$ of $(G,k,s)$ is \emph{compatible} with a vector $C = (C_1,\hdots,C_{k+1})\in\mathcal C$ if the strong component of $G - S$ containing $t_h$ is exactly $C_h$ for $h = 1,\hdots,k+1$.
%
%
For a given vector $C = (C_1,\hdots,C_{k+1})$, to determine whether a solution~$S$ of $(G,k,s)$ disjoint from $T$ and compatible with $C$ exists, we create several instances of the {\sc Skew Separator} problem.
To this end, note that if two sets $C_h,C_{h'}$ for distinct $t_h,t_h'\in T$ overlap, then actually $C_h = C_{h'}$ (and $t_h,t_h'\in C_h$).
So for each set $C_h$ we choose exactly one (arbitrary) \emph{representative $T$-vertex} among all $T$-vertices in~$C_h$ with consistent choice over overlapping (and thus equal) $C_h$'s.
Let $T'\subseteq T$ be the set of these representative vertices.
Now we generate precisely one instance $(G',\mathcal X_{\sigma'},\mathcal Y_{\sigma'},k)$ of {\sc Skew Separator} for each permutation~$\sigma'$ of~$T'$.
The graph $G'$ is the same in all these instances, and is obtained from~$G$ by replacing each unique set~$C_h$ by two vertices $t_h^+,t_h^-$ (where $t_h$ is the representative of~$C_h$), and connecting all vertices incoming to~$C_h$ in~$G$ by an in-arc to $t_h^+$ and all vertices outgoing from~$C_h$ in~$G$ by an arc outgoing from~$t_h^-$.
This way also arcs of the type $(t_j^-, t_h^+)$ are added but none of type $(t_j^-, t_h^-)$, $(t_j^+, t_h^-)$ or $(t_j^+, t_h^+)$.
Notice that this operation is well-defined and yields a simple digraph~$G'$, even if $t_{h'}\in C_h$ for some distinct $h,h'$.
The sets $\mathcal X_{\sigma'}$ and $\mathcal Y_{\sigma'}$ of ``sources'' and ``sinks'' depend on the permutation~$\sigma'$ with elements $\sigma'(1),\hdots,\sigma'(|T'|)$: let $\mathcal X_{\sigma'} = (t^-_{\sigma'(1)},\hdots,t^-_{\sigma'(|T'|)})$ and let $\mathcal Y_{\sigma'} = (t^+_{\sigma'(1)},\hdots,t^+_{\sigma'(|T'|)})$.

Thus, per triple $((G,k,s), T, C)$ we generate at most $|T'|! \leq |T|! = (k+1)!$ instances $(G',\mathcal X_{\sigma'},\mathcal Y_{\sigma'},k)$, the number of permutations of $T'$.

We now establish the correctness of this reduction, in the next two lemmas:
\begin{lemma}
\label{thm:correctnessreduceskewseparator}
  If an instance $(G,k,s)$ admits a solution $S$ disjoint from $T$, compatible with $C$ and for which $(t_{\sigma'(1)},\hdots,t_{\sigma'(|T'|)})$ is a topological order of the connected components of $G' - S$, then~$S$ forms a skew separator of size~$k$ for $(G,\mathcal X_{\sigma'},\mathcal Y_{\sigma'})$.
\end{lemma}
\begin{proof}
  Suppose, for the sake of contradiction, that the claim is false. 
  Then one of the two following cases must hold:
  \begin{itemize}
    \item For two vertices $t_h,t_{h'}\in T'$ with $\sigma'(h) < \sigma'(h')$, there would be a path~$P_1$ from~$t^-_{h'}$ to~$t^+_{h}$ in $G' - S$.
      This corresponds to a $C_{h'} \rightarrow C_{h}$ path in $G - S$.
      Either there is also a $C_{h} \rightarrow C_{h'}$ path in $G$ meaning that $C_{h} = C_{h'}$ in contradiction to our choice of $T'$ or the topological order $\sigma'$ of strong components was incorrect (as $C_{h'}$ must be before $C_h$)
    \item The in-vertex $t_h^-$ of the strong component containing $v_h$ would be reachable from the out-vertex $t_h^+$ of this strong component in the graph $G' - S$, because then the component would contain all the vertices on this path $P_2$ from $t_h^+$ to $t_h^-$, and by the way we constructed~$C_h$, the size of $|C_h\cup V(P_2)|$ in $G$ would be at least $s+1$, contradicting that the strong component of $G' - S$ containing $v_h$ has at most $s$ vertices.
		\qedhere
	\end{itemize}
\end{proof}

\begin{lemma}
\label{thm:smallskewsepimpliesdisjointcompatsoln}
  Conversely, if $S$ is a skew separator of $(G',\mathcal X_{\sigma'},\mathcal Y_{\sigma'})$ with size at most $k$, then~$S$ is a solution of $(G,k,s)$ disjoint from $T$ and compatible with~$C$.
\end{lemma}
\begin{proof}
  Suppose, for the sake of contradiction, that $S$ is not a solution for $(G,k,s)$.
  Then there is some strong component $Q$ in $G - S$ of size more than~$s$.
  By abuse of notation let $C =  \cup_{h=1}^{k+1} C_h$.
  Since neither $G[C]$ (by choice of $C$) nor $V(G) - C$ (as subdigraph of $G[V(G)\setminus T]$) contain strong components of size more than $s$, this component $Q$ must contain vertices from both~$C$ and $V(G)\setminus C$.
  Let $K$ be a closed walk of $Q$ that intersects both $G[C]$ and $G[V(G)\setminus C]$. Such a closed walk~$K$ must exist by $Q$ being strong.
	
  We consider two cases:
  \begin{itemize}
    \item The closed walk $K$ intersects a single unique component $C_h$.
      Then all other vertices of~$K$ are in $V(G)\setminus C_h$.
      Let $t_h$ be the representative of $C_h$.
      As $K$ intersects $G[V(G)\setminus C]$, $K$ leaves and enters $C_h$ at least once.
      This means that there is a walk $P_1$ in $G' - S$ that starts with the vertex~$t^-_h$ and ends with the vertex $t^+_h$, and all internal vertices of $P_1$ (of which there is at least one) are outside~$T'$.
      But this contradicts the assumption that $S$ is a skew separator for the tuple $(G',(t^-_{\sigma'(1)},\hdots,t^-_{\sigma'(|T'|)}),(t^+_{\sigma'(1)},\hdots,t^+_{\sigma'(|T'|)}))$ that should cut all walks from $t_h^-$ to~$t_h^+$.
    \item The closed walk $K$ intersects several different components $C_h$.
      Let $(C_{h_1},\hdots,C_{h_d},C_{h_1})$ be the order of components that we encounter when traversing along the walk $K$, starting from an arbitrary component $C_{h_1}$, where $d > 1$.
      Let $(t_{h_1},\hdots,t_{h_d},t_{h_1})$ be the corresponding representative vertices.
      Then there must be an index $j$ such that $h_j$ occurs after $h_{j+1\pmod{d+1}}$ in $(\sigma'(1), \sigma'(2), \dots \sigma'(|T'|)$.
      Hence in $G' - S$ is no path from $t^-_{h_j}$ to $t^+_{h_{j+1}}$.
      Now consider the subpath $P_2$ of $K$ that starts from the component $C_{h_j}$ and ends at component~$C_{h_{j+1}}$ and has its interior disjoint from both.
      Since all internal vertices on $P_2$ (by definition of~$P_2$) are not in any $C_h$, all such internal vertices of $P_2$ must be from $G - S - \cup_{i=1}^{|T'|}(X_i\cup Y_i)$, and the path $P_2$ corresponds to a path $P_2'$ in the graph $G' - S$ that starts from vertex $t^-_{h_j}$ and ends at vertex $t^+_{h_{j+1}}$.
      Again, this contradicts the assumption that $S$ is a skew separator for~$(G',\mathcal X_{\sigma'},\mathcal Y_{\sigma'})$.
  \end{itemize}
  Thus, the skew separator $S$ for $(G',\mathcal X_{\sigma'},\mathcal Y_{\sigma'})$ is a solution for $(G,k,s)$.
\end{proof}

In summary, we have reduced a single instance to the compression problem \textsc{Disjoint Bounded Size Strong Component Vertex Deletion Reduction} to at most $|\mathcal C|\cdot|T'|!$ instances $(G',\mathcal X_{\sigma'},\mathcal Y_{\sigma'},k)$ of the {\sc Skew Separator} problem, where each such instance corresponds to a permutation $\sigma'$ of $T'$.
The reduction just described implies that:
\begin{lemma}
  An input $(G,k,s,T)$ to the {\sc Disjoint Bounded Size Strong Component Vertex Deletion} problem is a ``yes''-instance if and only if at least one of the instances\linebreak $(G',\mathcal X_{\sigma'},\mathcal Y_{\sigma'},k)$ is a ``yes''-instance for the {\sc Skew Separator} problem.
\end{lemma}

So we invoke the algorithm of Proposition~\ref{thm:skewsepfpt} for each of the instances\linebreak $(G',\mathcal X_{\sigma'},\mathcal Y_{\sigma'},k)$.
If at least one of them is a ``yes''-instance then $(G,k,s,T)$ is a ``yes''-instance, otherwise $(G,k,s,T)$ is a ``no''-instance.
Hence, we conclude that {\sc Disjoint Bounded Size Strong Component Vertex Deletion Reduction} is fixed-parameter tractable with respect to the joint parameter $(k,s)$, and so is {\sc Bounded Size Strong Component Vertex Deletion}.
The overall run time of the algorithm is thus bounded by $|\mathcal C|\cdot|T'|!\cdot n^{\mathcal O(1)} \cdot 4^kkn^3 = (ks + s - 1)! \cdot (k+1)!\cdot 4^k \cdot n^{\mathcal O(1)} = 4^k(ks + k + s)!\cdot n^{\mathcal O(1)}$.
This completes the proof of Theorem~\ref{thm:boundedsizedeletion_main}.





\section{1-Out-Regular Arc (Vertex) Deletion}
\label{sec:inducedcycles}
In this section we give a fixed-parameter algorithm for the vertex deletion variant of Theorem~\ref{thm:2regulareuleriandeletion_main}.
Let $G$ be a digraph and let $k\in\mathbb N$.
A \emph{solution} for $(G,k)$ is a set $S$ of at most~$k$ vertices of $G$ such that every non-trivial strong component of $G - S$ is 1-out-regular.

We first apply the steps ``Iterative Compression'' and ``Disjoint Solution'' from Sect.~\ref{sec:generalsteps}.
This yields the {\sc Disjoint 1-Out-Regular Vertex Deletion Reduction} problem, where we seek a solution~$S$ of $(G,k)$ that is disjoint from and smaller than a solution~$T$ of $(G,k+1)$.

Then we continue with the technique of covering of shadows, as described in Sect.~\ref{sec:generalsteps}.
In our setting, let $\mathcal F$ be the collection of vertex sets of $G$ that induce a strong graph different from a simple directed cycle.
Then clearly $\mathcal F$ is $T$-connected and any solution~$S$ must intersect every such induced subgraph.

So we can use Proposition~\ref{thm:deterministiccoveringoftheshadow} on $(G, k, T)$ to construct sets $Z_1,\hdots,Z_t$ with $t \leq 2^{\mathcal{O}(k^2)}\log^2n$ such that one of these sets covers the shadow of our hypothetical solution $S$ with respect to $T$.
For each $Z_i$ we construct an instance, where we assume that $Z = Z_i \setminus T$ covers the shadow.
Note that a vertex of $T$ is never in the shadow.
As we assume that $Z \cup T$ is disjoint from a solution, we reject an instance if $G[Z \cup T]$ contains a member of $\mathcal{F}$ as a subgraph.
\begin{observation}
\label{obs:shadowcontainsnononcycle}
  $G[Z \cup T]$ has no subgraph in $\mathcal{F}$.
\end{observation}


Normally, one would give a ``torso'' operation which transforms $(G,k)$ with the use of $Z$ into an instance $(G',k')$ of the same problem which has a shadowless solution if and only if the original instance has any solution.
Instead, our torso operation reduces to a similar problem while maintaining solution equivalence.

\medskip
\noindent
\textbf{Reducing the Instance by the Torso Operation.}
Our torso operation works directly on the graph.
It reduces the original instance to one of a new problem called {\sc Disjoint Shadow-less Good 1-Out-Regular Vertex Deletion Reduction}; afterwards we show the solution equivalence.

\begin{definition}
\label{defn:torso}
  Let $(G,T,k)$ be an instance of {\sc Disjoint 1-Out-Regular Vertex Deletion Reduction} and let $Z\subseteq V(G)$.
  Then $\mathsf{torso}(G,Z)$ defines the digraph with vertex set $V(G)\setminus Z$ and \emph{good} and \emph{bad} arcs.
  An arc $(u,v)$ for $u,v \not\in Z$ is introduced whenever there is an $u \rightarrow v$ path in $G$ (of length at least 1) whose internal vertices are all in $Z$.
  We mark $(u, v)$ as \emph{good} if this path $P$ is unique and there is no cycle $O$ in $G[Z]$ with $O \cap P \not = \emptyset$.
  Otherwise, we mark it as a \emph{bad} arc.
\end{definition}
Note that every arc between vertices not in $Z$ also forms a path as above.
Therefore $G[V(G) \setminus Z]$ is a subdigraph of $\mathsf{torso}(G,Z)$.
Also, $\mathsf{torso}(G,Z)$ may contain self-loops at vertices~$v$ from cycles with only the vertex $v$ outside of~$Z$.
In $\mathsf{torso}(G,Z)$, we call a cycle \emph{good} if it consists of only good arcs.
(A non-good cycle in $\mathsf{torso}(G,Z)$ can contain both good arcs and bad arcs.)

Now we want to compute a vertex set of size $k$ whose deletion from $G' = \mathsf{torso}(G,Z)$ yields a digraph whose every non-trivial strong component is a cycle of good arcs.
We call this problem {\sc Disjoint Shadow-less Good 1-Out-Regular Vertex Deletion Reduction}.
To simplify notation we construct a set $\mathcal{F}_\textsf{bad}$ which contains all strong subdigraphs of $G$ that are not trivial or good cycles.
Then $S$ is a solution to $G'$ if and only if $G' - S$ contains no subdigraph in $\mathcal{F}_\textsf{bad}$.
In the next lemma we verify that our new problem is indeed equivalent to the original problem, assuming that there is a solution disjoint from $Z$.

\begin{lemma}[torso preserves obstructions]
\label{thm:torsopreservesobstructions}
  Let $G$ be a digraph, $T, Z\subseteq V(G)$ as above and $G' = \mathsf{torso}(G,Z)$.
  For any $S \subseteq V(G) \setminus (Z \cup T)$ it holds that $G - S$ contains a subdigraph in~$\mathcal{F}$ if and only if $G' - S$ contains a subdigraph in $\mathcal{F}_\textsf{bad}$.
\end{lemma}
\begin{proof}
  In the forward direction, if $G' - S$ contains a subgraph $F' \in \mathcal{F}_\textsf{bad}$, we can replace the arcs of $F'$ as follows:
  All good arcs are replaced by their unique path in the torso operation.
  For a bad arc $(x,y)$ we insert all $x \rightarrow y$-paths whose internal vertices completely belong to $Z$.
  If there is only a single such path $P$ then by definition there is a cycle $O$ in~$G[Z]$ that intersects~$P$.
  We also insert all cycles $O$ of this type.
  Call the resulting graph~$F$.
  
  This digraph $F$ is a subdigraph of $G - S$ and is strong, as $F'$ was strong and all added vertices have a path from and to $V(F')$.
  Now, either $F'$ was not a cycle, then $F$ is also not a cycle or it contained a bad arc and we have inserted at least two parallel paths or a cycle.
  In any case, we have $F \in \mathcal{F}$.
	
	\looseness=-1
  In the backward direction, let $G - S$ have a subdigraph $F \in \mathcal{F}$.
  Assume for contradiction that $G' - S$ has no subdigraph in~$\mathcal{F}_\textsf{bad}$.
  We will show that $F' = \mathsf{torso}(F, Z) \in \mathcal{F}_\textsf{bad}$.
  Note that the torso operation preserves subdigraph relations and connection.
  By Observation~\ref{obs:shadowcontainsnononcycle} we know that there is a $v \in V(F) \setminus (Z \cup T)$.
  Furthermore, we know that there is also a $t \in V(F) \cap T$ as~$T$ is a solution to $G$.
  From $Z \cap T = \emptyset$ by definition we know that $v, t \not \in Z$ and hence in $V(F')$.
  As~$F$ is strong, there is a closed walk $O$ through $v$ and~$t$ in~$F$.
	
  \begin{claim}
  \label{thm:oisacycle}
    $O$ is a cycle.
  \end{claim}
  \begin{proof}[Proof of Claim~\ref{thm:oisacycle}]
  \renewcommand{\qedsymbol}{$\diamond$}
    Suppose, for sake of contradiction, that $O$ is not a cycle.
	Let $w$ a vertex that is visited at least twice when traversing~$O$.
	Let $x_1, w, y_1$ be the first traversal and $x_2, w, y_2$ be the second one.
	Without loss of generality, we can assume that $x_1, x_2, y_1, y_2 \not\in Z$ by replacing them by the next vertex outside of $Z$.
	If $w \in Z$ then the arcs $(x_1, y_1), (x_1, y_2), (x_2, y_1), (x_2, y_2)$ all exist in the strong subdigraph $\mathsf{torso}(O, Z)$.
	Therefore, $\mathsf{torso}(O,Z) \in \mathcal{F}_\textsf{bad}$ in contradiction to the fact that $\mathsf{torso}(O,Z)$ is a subdigraph of $G' - S$.
	Else, arcs $(x_1,w), (x_2,w), (w, y_1), (w, y_2)$ would exist, giving the same contradiction.
    This completes the proof of the claim.
  \end{proof}
	
  Now $F$ is strong and not a cycle, and therefore has to contain a (possibly closed) $x \rightarrow y$ path $R$ with the following properties:
  \begin{itemize}
  	\item $x, y \in V(O)$,
  	\item $R$ contains no arc from $O$,
  	\item all internal vertices of $R$ are disjoint of $V(O)$.
  \end{itemize}  
  
  Then there are paths $O_x$ and $O_y$ in $O$ such that their endpoints are not in $Z$ but all their interior vertices and furthermore $x \in V(O_x)$ respectively $y \in V(O_y)$.
  If $x \not \in Z$ (resp. $y \not \in Z$), set $x_1 = x_2 = x$ (resp. $y_1 = y_2 = y)$.
	
  If $R$ contains some interior vertex $u \not\in Z$, the path $O_x[x_1, x] \circ R[x, u]$ is in~$F$ and shrinks to a $x_1 \rightarrow u$-path in $F'$.
  As $u \notin V(O)$ we get that $x_1$ has at least two out-arcs $(x_1, x_2),(x_1, u)$ in $F'$ and therefore $F' = \mathsf{torso}(F, Z) \in \mathcal{F}_\textsf{bad}$, a contradiction.
  Thus, the interior of $R$ lies in $Z$.
  Furthermore, if $(x_1, x_2) \not= (y_1, y_2)$ then $O_x[x_1, x] \circ R \circ O_y[y, y_2]$, is a $x_1 \rightarrow y_2$ path in $F$.
  Note that $x_2\not=y_2$ as $O$ is a cycle and $(x_1,x_2)\not = (y_1,y_2)$.
  Therefore the path is shrunk by the torso operation to the arc $(x_1,y_2)$.
  But then $x_1$ has two outgoing arcs in $F'$ and as $F'$ is still strong, $F' = \mathsf{torso}(F, Z) \in \mathcal{F}_\textsf{bad}$.
  Therefore, we have $(x_1, x_2) = (y_1, y_2)$ and also $O_x = O_y$ as otherwise the arc would be bad (because there are two different $x_1 \rightarrow x_2$-paths).
  If $x$ lies before $y$ on~$O_x$ the path $P = O_x[x_1, x] \circ R \circ O_x[y, x_2]$ is a $x_1 \rightarrow x_2$-path in $F$.
  As the interior of $O_x$ and~$R$ is in~$Z$ this would give a second $x_1 \rightarrow x_2$-path, making $(x_1, x_2)$ bad.
  The last case is if $y$ lies before $x$ on $O_x$.
  Then $R \circ O_x[y,x]$ forms a cycle in $Z$ which intersects $O_x$ at least in the vertex~$x$, again proving that $(x_1, x_2)$ should be bad.
\end{proof}

The above lemma shows that $S$ is a solution of an instance $(G, T, k)$ for {\sc Disjoint 1-Out-Regular Vertex Deletion Reduction} disjoint of $Z$ if and only if it is a solution of $(\mathsf{torso}(G, Z), T, k)$ for {\sc Disjoint Shadow-less Good 1-Out-Regular Vertex Deletion Reduction}.
As connections between vertices are preserved by the torso operation and the torso graph contains no vertices in $Z$, we can reduce our search for $(\mathsf{torso}(G, Z), T, k)$ to shadow-less solutions (justifying the name).

\medskip
\noindent
\textbf{Finding a Shadowless Solution.}
Consider an instance $(G, T, k)$ of {\sc Disjoint Shadow-less Good 1-Out-Regular Vertex Deletion Reduction}.
Normally, after the torso operation a pushing argument is applied.
However, we give an algorithm that recovers the last connected component of $G$.
As $T$ is already a solution, but disjoint of the new solution $S$, we take it as a starting point of our recovery.
Observe that, without loss of generality, each vertex $t$ in $T$ has out-degree at least one in $G - T\setminus\{t\}$, for otherwise already $T - t$ is a solution.

Consider a topological order of the strong components of $G - S$, say $C_1,\hdots,C_\ell$, i.e., there can be an arc from $C_i$ to $C_j$ only if $i < j$.
We claim that the last strong component $C_\ell$ in the topological ordering of $G - S$ contains a non-empty subset~$T_0$ of $T$.
For if $C_\ell$ did not contain any vertex from $T$, then the vertices of $C_\ell$ cannot reach any vertex of $T$, contradicting that~$S$ is a shadowless solution of $(G,k)$.

Since $T_0$ is the subset of $T$ present in $C_\ell$ and arcs between strong components can only be from earlier to later components, we have that there are no outgoing arcs from $C_\ell$ in $G-S$.

We guess a vertex $t$ inside $T_0$.
This gives $|T| \leq k + 1$ choices for $t$.
For each guess of $t$ we try to find the component $C_\ell$, similarly to the bounded-size case.
The component $C_\ell$ will either be trivial or not.

If $C_\ell$ is a trivial component, then $V(C_\ell) = \{t\}$, and so we delete all out-neighbors of~$t$ in $G - T$ and place them into the new set $S$.
Hence, we must decrease the parameter $k$ by the number of out-neighbors of $t$ in $G - T$, which by assumption is at least one.

Else, if the component $C_{\ell}$ is non-trivial, define $v_0 = t$ and notice that exactly one out-neighbor~$v_1$ of $v_0$ belongs to $C_\ell$.
Set $i = 0$ and notice that every out-neighbor of $v_i$ other than~$v_{i+1}$ must be removed from the graph~$G$ as $C_\ell$ is the last component in the topological ordering of $G - S$, there is no later component where those out-neighbors could go.
This observation gives rise to a natural branching procedure: we guess the out-neighbor $v_{i+1}$ of $v_i$ that belongs to $C_\ell$ and remove all other out-neighbors of $v_i$ from the graph.
We then repeat this branching step with $i\mapsto i+1$ until we get back to the vertex~$t$ of $T_0$ we started with.
This way, we obtain exactly the last component~$C_\ell$, forming a cycle.
This branching results in at least one deletion as long as $v_i$ has out-degree at least two.
If the out-degree of $v_i$ is exactly one, then we simple proceed by setting $v_i := v_{i+1}$ (and increment $i$).
In any case we stop early if $(v_i, v_{i+1})$ is a bad arc, as this arc may not be contained in a strong component.

Recall that the vertices $t = v_0,v_1,\hdots$ must \emph{not} belong to $S$, whereas the deleted out-neighbors of $v_i$ must belong to $S$.
From another perspective, the deleted out-neighbors of $v_i$ must \emph{not} belong to $T$.
So once we reached back at the vertex $v_j = t$ for some $j\geq 1$, we have indeed found the component $C_\ell$ that we were looking for.

Let us shortly analyze the run time of the branching step.
As for each vertex~$v_i$, we have to remove all its out-neighbors from $G$ except one and include them into the hypothetical solution~$S$ of size at most~$k$, we immediately know that the degree of $v_i$ in $G$ can be at most $k + 1$.
Otherwise, we have to include~$v_i$ into $S$.
Therefore, there are at most $k+1$ branches to consider to identify the unique out-neighbor $v_{i+1}$ of $v_i$ in $C_{\ell}$.
So for each vertex $v_i$ with out-degree at least two we branch into at most $k+1$ ways, and do so for at most $k$ vertices, yielding a run time of $O((k+1)^k)$ for the entire branching.

Once we recovered the last strong component $C_\ell$ of $G - S$, we remove the set $V(C_\ell)$ from~$G$ and repeat: we then recover $C_{\ell-1}$ as the last strong component, and so on until $C_1$.

\medskip
\noindent
\textbf{Algorithm for Disjoint 1-Out-Regular Vertex Deletion Reduction.}\linebreak
Lemma~\ref{thm:torsopreservesobstructions} and the branching procedure combined give a bounded search tree algorithm for {\sc Disjoint 1-Out-Regular Vertex Deletion Reduction}:

\begin{enumerate}[labelindent=2em,labelwidth=2pt,label=\underline{Step~\arabic*.},itemindent=3.5em,leftmargin=0pt]
  \renewcommand{\theenumi}{Step~\arabic{enumi}}
  \item\label{itm:step1} For a given instance $I = (G,T,k)$, use Proposition~\ref{thm:deterministiccoveringoftheshadow} to obtain a set of instances $\{Z_1,\hdots,Z_p\}$ where $p \leq 2^{\mathcal O(k^2)}\log^2n$, and Lemma~\ref{thm:torsopreservesobstructions} implies
  \begin{itemize}
    \item If $I$ is a ``no''-instance then all reduced instances $(\mathsf{torso}(G,Z_j), T, k)$ are ``no''-instances, for $j = 1,\hdots,p$.
    \item If $I$ is a ``yes''-instance then there is at least one $i\in\{1,\hdots,p\}$ such that there is a solution $S^\star$ for $I$ which is a shadowless solution for the reduced instance $(\mathsf{torso}(G,Z_i), T, k)$.
  \end{itemize}
  So at this step we branch into $p \leq 2^{\mathcal{O}(k^2)}\log^2n$ directions.
  \item\label{itm:step2} For each of the instances obtained from \ref{itm:step1}, recover the component~$C_\ell$ by guessing the vertex $t = v_0$. Afterwards, recover $C_{\ell-1},\hdots,C_1$ in this order.
    
  So at this step we branch into at most $\mathcal{O}(k\cdot(k+1)^k)$ directions.
\end{enumerate}
%
We then repeatedly perform \ref{itm:step1} and \ref{itm:step2}.
Note that for every instance, one execution of \ref{itm:step1} and \ref{itm:step2} gives rise to $2^{\mathcal{O}(k^2)}\log^2n$ instances such that for each instance, we either know that the answer is ``no'' or the budget~$k$ has decreased, because each important separator is non-empty.
Therefore, considering a level as an execution of \ref{itm:step1} followed by \ref{itm:step2}, the height of the search tree is at most~$k$.
Each time we branch into at most\linebreak $2^{\mathcal{O}(k^2)}\log^2n\cdot \mathcal{O}(k\cdot (k+1)^k)$ directions.
Hence the total number of nodes in the search tree is
\begin{eqnarray*}
  \left(2^{\mathcal{O}(k^2)}\log^2n\right)^k\cdot \mathcal{O}\left(k\cdot (k+1)^k\right) &= &\left(2^{\mathcal{O}(k^2)}\right)^k\left(\log^2n\right)^k  \cdot \mathcal{O}\left((k+1)^{k+1}\right)\\
  & = &    2^{\mathcal{O}(k^3)} \cdot 2^{\mathcal{O}(k \log k)}\left(\log^2n\right)^k\\
  & = & 2^{\mathcal{O}(k^3)}\cdot \mathcal{O}\left(((2k\log k)^k + n/2^k)^3\right)\\
  & = &  2^{\mathcal{O}(k^3)}\cdot \mathcal{O}(n^3) \enspace .
\end{eqnarray*}
We then check the leaf nodes of the search tree and see if there are any strong components other than cycles left after the budget $k$ has become zero.
If for at least one of the leaf nodes the corresponding graph only has strong components that are cycles then the given instance is a ``yes''-instance.
Otherwise, it is a ``no''-instance.
This gives an $2^{\mathcal{O}(k^3)}\cdot n^{\mathcal{O}(1)}$-time algorithm for {\sc Disjoint 1-Out-Regular Vertex Deletion Reduction}.
So overall, we have an $2^{\mathcal{O}(k^3)}\cdot n^{\mathcal{O}(1)}$-time algorithm for the {\sc 1-Out-Regular Vertex Deletion} problem.

\section{Polynomial Parameter Transformations Between Arc Deletion and Vertex Deletion}
\label{sec:ppts}
In this section we prove the existence of polynomial parameter transformations between {\sc Bounded Size Strong Component Arc Deletion} and {\sc Bounded Size Strong Component Vertex Deletion} in both directions, as well as a polynomial parameter transformation from {\sc 1-Out-Regular Arc Deletion} to {\sc 1-Out-Regular Vertex Deletion}.
These complete the proofs of Theorem~\ref{thm:boundedsizedeletion_main} and Theorem~\ref{thm:2regulareuleriandeletion_main}.

\subsection{Bounded Size Strong Component Deletion:\PDFcompatibleLineBreak Polynomial Parameter Transformation from Arc to Vertex Version}
\label{sec:boundedsizereductionarc2vertex}
Our first transformation reduces an instance of {\sc Bounded Size Strong Component Arc Deletion} to an instance of {\sc Bounded Size Strong Component Vertex Deletion}.
This completes the proof of Theorem~\ref{thm:boundedsizedeletion_main}.
While our transformation keeps the parameter $k$ constant, the parameter~$s$ increases to $(k + 1) s^3$.
This is due to a replacement of all vertices by complete graphs of size roughly $ks^2$, therefore increasing the size of eligible components.

\begin{lemma}
	\label{thm:boundedsize_arctovertex}
	Given an instance $(G,k,s)$ of {\sc Bounded Size Strong Component Arc Deletion} we can compute in polynomial time a solution-wise equivalent instance $(G', k', s')$ of {\sc Bounded Size Strong Component Vertex Deletion} with $k' = k$ and $s' =  (k + 1) s^3$.
\end{lemma}
\begin{proof}
We first bound the number of parallel arcs in $G$.
Note that if there are more than $k+1$ arcs between a pair of vertices running in the same direction, we can remove additional arcs as at least one of these arcs remains after the removal of $k$ arcs.
Thus we can restrict ourselves to instances with at most $k+1$ parallel arcs per ordered vertex pair.
In such digraphs any subdigraph with at most $s$ vertices has at most $s_a := (k+1) s (s-1)$ arcs.
The idea is now to subdivide the arcs by a vertex and replace the original vertices by complete directed graphs of size $s_a + k + 1$.
Then the addition of a original vertex to a strong component has more impact than the artificial vertices needed to subdivide the arcs.
Formally, we define our new digraph~$G'$ as follows:
\begin{eqnarray*}
  V(G') & = & \{v_i~|~v\in V(G), 1 \leq i \leq s_a + k + 1\} \cup\{u_a~|~a\in A(G)\},\\
  A(G') & = & \{(v_i,v_j)~|~v\in V(G), 1 \leq i,j \leq s_a + k + 1, i \not = j\}\\
  &&\cup \{(v_i,u_a)~|~a = (v,w) \in A(G), 1\leq i \leq s_a + k +1\}\\
  &&\cup\{(u_a,w_i)~|~a = (v,w) \in A(G), 1\leq i \leq s_a + k +1\} .
\end{eqnarray*}
Finally we set $s' = s (s_a + k + 1) + s_a = (k + 1)s^3$ and get the resulting instance $(G', k, s')$ of {\sc Bounded Size Strong Component Vertex Deletion}.
It remains to show that the two instances are indeed solution-wise equivalent.

For the forward direction, let $S$ be a set of at most $k$ arcs such that every strong component of $G-S$ has at most $s$ vertices.
Let $S' = \{u_a~|~a\in S\}$.
By construction we have that $G' - S'$ is equivalent to applying above transformation to the graph $G - S$.
As our transformation preserves connectedness we have a one to one correspondence between strong components of $G' - S'$ and $G - S$.
Let $X'$ be a strong component of $G' - S'$ and $X$ its corresponding set in $G -S$.
We know that $E(G[X])$ has size at most $(k+1) |X| (|X| - 1) \leq s_a$.
Thus, $X'$ contains at most~$s_a$ vertices of type~$u_a$.
Furthermore, there are at most $|X|(s_a + k + 1)  \leq s (s_a + k + 1)$ vertices of type $v_i$.
Hence, we have $|X'| \leq s_a + s (s_a + k + 1) = s'$, and by $|S'| = |S| \leq k$ we know that $S'$ is a valid solution to $(G', k, s')$.

For the reverse direction, let $S'$ be a set of at most $k$ vertices such that every strong component of $G'-S'$ has at most $s'$ vertices.
Then, we claim that the set \mbox{$S = \{a \in A(G)~|~u_a \in S'\}$} is a solution to $(G,k,s)$.
Obviously, $|S| \leq |S'| \leq k$.
We now want to show that the strong components in $G-S$ do contain at most~$s$ vertices.
As $s_a + k +1 > k$ we know that for every $v \in V(G)$ at least one $v_i$ remains in $G'-S'$.
Because all~$v_i$ have the same neighbors, removing the vertices of type $v_i$ from $S'$ does not change connectivity of $G - S'$.
Now again there is a one-to-one correspondence between the strong components of $G-S$ and~$G' - S'$.
The strong components of $G' - S'$ are missing at most $k$ vertices of type~$v_i$ which are in $S'$.
Let $X$ be a strong component in~$G' - S'$.
Let $W \subset V(G)$ be the set of all vertices $w \in V(G)$ in $G$ such that $X$ contains a vertex~$w_i$.
If $|W| > s$ then $X$ contains at least $(s_a + k + 1)|W| - k \geq (s_a + k + 1) s + s_a + k + 1 - k = s' + 1$ vertices, a contradiction to the fact that~$S'$ was solution for $(G', k, s')$.
Thus, $|W| \leq s$ and by the one to one correspondence of strong components, we know that $W$ is indeed a strong component of $G - S$.
As~$X$ was chosen arbitrarily and all strong components of $G - S$ have a counterpart in $G' - S'$, this completes the proof.
\end{proof}

\subsection{Bounded Size Strong Component Deletion:\PDFcompatibleLineBreak Polynomial Parameter Transformation from Vertex to Arc Version}
\label{sec:boundedsizereductionvertex2arc}
Here we state a transformation from {\sc Bounded Size Strong Component Vertex Deletion} to {\sc Bounded Size Strong Component Arc Deletion}.
This transformation is not needed for any theorem, but we state it here nonetheless for completeness.
Note that, unlike the reduction in backwards direction, the parameter increase of $s$ is only linear and does not depend on $k$.

\begin{lemma}
\label{thm:boundedsize_vertextoarc}
  Given an instance $(G,k,s)$ of {\sc Bounded Size Strong Component Vertex Deletion} we can compute in polynomial time a solution-wise equivalent instance $(G', k', s')$ of {\sc Bounded Size Strong Component Arc Deletion} with $k' = k$ and $s' = 2s$.
\end{lemma}
\begin{proof}
  Given an instance $(G,k,s)$ of {\sc Bounded Size Strong Component Vertex Deletion}, create a digraph $G'$ from $G$ by splitting each vertex $v\in V(G)$ into two vertices $v^+,v^-$, adding the arc from~$v^-$ to $v^+$, and connecting all in-neighbors $u$ of $v$ in $G$ by $k + 1$ parallel arcs from $u^+$ to $v^-$ in~$G'$, and all out-neighbors $u$ of $v$ in $G$ by $k + 1$ parallel arcs from $v^+$ to $u^-$ in~$G'$.
  In other words, we set
  \begin{eqnarray*}
  V(G') & = & \{v^+,v^-~|~v\in V(G)\},\\
  A(G') & = & \{(v^-,v^+)~|~v\in V(G)\}\\
  &   &\cup \{(u^+,v^-)^{k+1}~|~u\in N_G^-(v),v\in V(G)\}\\
  &   &\cup \{(v^+,u^-)^{k+1}~|~u\in N_G^+(v),v\in V(G)\} \enspace .
\end{eqnarray*}
We further set $k' = k$ and $s' = 2s$.
Then $(G',k',s')$ is an instance of {\sc Bounded Size Strong Component Arc Deletion}.
It remains to check solution-wise equivalence.

In the forward direction, let $S$ be a set of at most $k$ vertices such that in $G - S$ every strong component has at most $s$ vertices.
Let $S' = \{(v^-,v^+)~|~v\in S\}$ be the corresponding set of $k$ arcs in $G'$.
The number of vertices in each strong component of $G' - S'$ is now exactly twice the number of vertices in its corresponding component in $G - S$. Therefore, every strong component of $G' - S'$ consists of at most $s' = 2s$ vertices.

In the backward direction, let $S'$ be a set of at most $k$ arcs such that in $G' - S'$ every strong component has at most $s' = 2s$ vertices.
We first argue that we can change $S'$ in such a way that it will only consist of arcs of the form $(v^-,v^+)$ for some vertex $v\in V(G)$.
This is clear if we use the trick with $k + 1$ parallel arcs.
Else, we have to argue as follows:
For suppose there is an arc $(v^+,u^-)\in S'$ for some vertices $v^+,u^-\in V(G')$ corresponding to distinct vertices $v,u\in V(G)$.
Then for $S'' = S'\setminus\{(v^+,u^-)\}\cup\{(v^-,v^+)\}$, vertices $v^+,u^-$ do not belong to the same strong component of $G' - S''$ since~$v^+$ is a source of $G' - S''$.
Therefore, for each vertex $v^{\pm}\in V(G')$ the size of the (uniquely determined) strong component in $G - S''$ containing $v^{\pm}$ is at most the size of the strong component in $G' - S'$ containing $v^{\pm}$.
This justifies the assumption that $S' = \{(v^-,v^+)~|~v\in V(G)\}$.)
Now let $S = \{v\in V(G)~|~(v^-,v^+)\in S'\}$ be the set of at most $k$ vertices in $G$ corresponding to the arcs in $S'$.
The number of vertices in each strong component of $G - S$ is now exactly half the number of vertices in its corresponding component in $G' - S'$.
Therefore, every strong component of $G - S$ consists of at most $s = s'/2$ vertices.
\end{proof}

\subsection{1-Out Regular Deletion:\PDFcompatibleLineBreak Polynomial Parameter Transformation from Arc to Vertex Version}
\label{sec:1outregulararctovertex}
Last but not least, we show a transformation from {\sc 1-Out-Regular Arc Deletion} to {\sc 1-Out-Regular Vertex Deletion}.
Thus, by the fixed-parameter tractability of the vertex deletion version as shown in Theorem~\ref{thm:2regulareuleriandeletion_main}, we obtain fixed-parameter tractability of the arc deletion version.
Note that this reduction, unlike the others is parameter preserving.

\begin{lemma}
	\label{thm:oneoutregular_arctovertex}
	Given an instance $(G,k)$ of {\sc 1-Out-Regular Arc Deletion}  we can compute in polynomial time a solution-wise equivalent instance $(G', k')$ of {\sc 1-Out-Regular Vertex Deletion} with $k' = k$.
\end{lemma}

\begin{proof}
Let $G'$ be the directed line graph of $G$, that is $G'$ has a vertex~$v_a$ for every arc $a \in A(G)$ and the arc $(v_a, v_b)$ exists in $G'$ if and only if $a = (u,v)\in A(G)$ and $b = (v,w)\in A(G)$ for some $u,v,w \in V(G)$.

Obviously, there is a one-to-one correspondence between arcs in $G$ and vertices in $G'$.
This also holds if there is a set $S$ of arcs in $G$ and $S'$ its corresponding set of vertices in~$G'$ for the digraphs $G - S$ and $G' - S'$.
The correspondence also holds for non-trivial strong components as a closed walk on vertices $v_1, \dots, v_t$ and arcs $a_1, \dots, a_t$ corresponds to a closed walk on the vertices $v_{a_1}, \dots, v_{a_t}$ in $G'$.
We now show the solution-wise equivalence of the instances.

For the forward direction, let $S$ be a solution to $(G, k)$.
Let $S' = \{v_a~|~a \in S\}$.
As $|S'| = |S| \leq k$, our candidate fulfills the size bound.
Let now $X'$ be a strong component of $G' - S'$.
Assume for contradiction that $X'$ is neither trivial nor $1$-out-regular.
By above correspondence there is a non-trivial strong component $X$ in $G - S$ that has the arcs which~$X'$ possesses as vertices.
As $S$ is a solution to $(G, k)$, $X$ is $1$-out-regular (as it is not trivial).
Therefore, $G[X]$ forms a cycle $O$.
This cycle has a correspondent cycle $O'$ in $G'[X']$.
Since $O$ visits all arcs of $G[X]$, $O'$ is a Hamiltonian cycle for~$G'[X']$.
As $G'[X']$ is not $1$-out-regular, there must be an arc $(v_a, v_b) \in E(G'[X'])$ which is not part of~$O'$.
This arc means that the arcs~$a$ and~$b$ share a vertex~$v$ in~$G[X]$ albeit being not adjacent in $O$.
Thus,~$v$ has out-degree at least two in~$G[X]$, a contradiction.
Therefore, $S'$ is a solution to $(G', k)$.

For the reverse direction, let $S'$ be a solution to $(G',k)$.
Let $S = \{a \in A(G)~|~v_a \in S'\}$.
Again we have $|S| = |S'| \leq k$ and thus the size bound fulfilled.
Let $X$ be a strong component in~$G - S$.
Assume, for sake of contradiction, that $X$ is neither trivial nor $1$-out-regular.
This means that $G[X]$ contains a cycle $O_X$ and a (possibly closed) walk $P$ with both endpoints on~$O_X$ and its interior disjoint of it.
Let $x$ be the start vertex of $P$ and $a$ the first arc of~$P$.
Furthermore, let $b = (v, x)$ and $c = (x,w)$ be the arcs adjacent to $x$ on $O$.
Then $G' - S'$ contains the vertices $v_a, v_b, v_c$, and by preservation of strong connectivity they are in the same connected component of $G' - S'$.
But by choice of $a,b,c$ the arcs $(v_b, v_a)$ and $(v_b, v_c)$ exist in $G' - S'$.
This means that $v_b$ has out-degree at least two in its strong component in $G' - S'$, a contradiction.
In conclusion,~$S$ must be a solution for $(G, k)$.
\end{proof}

\bibliographystyle{abbrvnat} 
\bibliography{main}

\end{document}